\documentclass[12pt]{article}
\usepackage[left=2.5cm,right=2cm,top=3cm,bottom=2cm]{geometry}
\usepackage{amsmath, amssymb, amsthm, textcomp} 
\usepackage{comment, fancyhdr,comment}
\usepackage{graphicx,color}
\usepackage{bm}
\usepackage{appendix}
\usepackage{upgreek}

\usepackage{mathrsfs}
\usepackage{array, amsfonts, mathrsfs}

\newtheorem{lemma}{Lemma}
\newtheorem{theorem}{Theorem}
\newtheorem{definition}{Definition}
\newtheorem{corollary}{Corollary}
\newtheorem{proposition}{Proposition}

\def\mE{{\mathscr E}}
\def\fH{{\mathfrak H}}

\title{On stability of determination of Riemann surface from its DN-map.}
\author{M.I.Belishev\thanks {St.Petersburg Department of Steklov Mathematical Institute, St.Petersburg, Russia,
        \newline
        e-mail: belishev@pdmi.ras.ru,
        \newline
        ORCID: 0000-0002-4759-7428;
        \newline
        supported by RFBR grant 20-01-00627-a
        \newline }.\,
        D.V.Korikov\thanks {St.Petersburg Department of Steklov Mathematical Institute, St. Petersburg, Russia,
        \newline
        e-mail: thecakeisalie@list.ru,
        \newline
        ORCID: 0000-0002-3212-5874;
        \newline
        supported by RFBR grant 20-01-00627-a
        \newline }.}

\date{}

\begin{document}
\maketitle
\begin{abstract}
Suppose that $M$ is a Riemann surface with boundary
$\partial M$,
$\Lambda$ is its DN-map, and $\mathscr E:M\to\mathbb{C}^{n}$ 
is a holomorphic immersion. Let $M'$ be diffeomorphic to $M$,
$\partial M=\partial M'$; let $\Lambda'$ be the DN map of $M'$.
Let us write $M'\in\mathbb M_t$ if
$\parallel\Lambda'-\Lambda\parallel_{H^{1}(\partial M)\to
L_{2}(\partial M)}\leqslant t$ holds. We show that, for any
holomorphic immersion $\mE: M \to \mathbb C^n$ {\rm (}$n\geqslant
1${\rm)}, the relation
\begin{equation*}
\sup_{M'\in
\mathbb{M}_{t}}\inf_{\mE'}d_{H}(\mE'(M'),\mE(M))\underset{t\to
0}{\longrightarrow}0,
\end{equation*}
holds, where $d_H$ is the Haussdorf distance in $\mathbb C^n$ and
the infimum is taken over all holomorphic immersions $\mE':
M'\mapsto\mathbb C^n$.
\end{abstract}

\noindent{\bf Key words:}\,\,\,electric impedance tomography of
surfaces, holomorphic immersions, Di\-rich\-let-to-Neumann map,
stability of determination.

\noindent{\bf MSC:}\,\,\,35R30, 46J15, 46J20, 30F15.
\bigskip

\section{Introduction}\label{sec Introduc}

\subsubsection*{Statement}

$\bullet \ $ Let $(M,g)$ be a smooth\footnote{throughout the
paper, {\it smooth} means $C^\infty$-smooth} compact
two-dimensional Riemann manifold ({\it surface}) with the smooth
boundary $(\Gamma,dl)$, $g$ is the smooth metric tensor, $\Gamma$
is diffeomorpic to a circle, and $dl$ is the length element on
$\Gamma$ induced by the metric $g$. By $\nu$ we denote the outward
normal to $\Gamma$.

The {\it Dirichlet-to-Neumann map}
(DN-map) of the surface 
acts on smooth functions on $\Gamma$ by the rule $\Lambda
f:=\partial_{\nu}u^{f}|_{\Gamma}$, where $u^{f}$ 
satisfies
 $$
\Delta_g u=0\quad\text{in}\,\,M\setminus\Gamma,\qquad
u|_{\Gamma}=f.
 $$
A specific feature of the two-dimensional case is the following.
Let the metrics $g$ and $\rho g$ be conformal equivalent via a
smooth $\rho>0$ satisfying $\rho|_\Gamma=1$; then the DN-maps of
$(M,g)$ and $(M,\rho g)$ do coincide.

The question that is known as {\it electric impedance tomography
problem }(EIT), can be posed as follows: to what extend does the
DN-map determine the surface?

\smallskip

\noindent$\bullet$\,\,\,Traditional understanding of `to solve an
inverse problem {\it completely}' includes solving several
problems:

{\bf i}\,\,\,\,\,\,to establish a relevant uniqueness of
determination,

{\bf ii}\,\,\,\,\,to provide a procedure that recovers the object
under determination,

{\bf iii}\,\,\,to provide a characteristic description of the
inverse data that ensures the solvability of the problem,

{\bf iv}\,\,\,to study a stability of determination that is to
analyze how the variations of data influent on the solution.

Problems {\bf i} and {\bf ii} were first solved in
\cite{LU}: it was shown that $\Lambda$ determines $(M,g)$ up to
conformal equivalence. Such a result is consistent with the
specific feature of the 2-dim problem EIT, which is mentioned
above. Then, in \cite{B} this result was obtained by algebraic
version of the boundary control method (BCM) \cite{B UMN}.
Recently, the it was extended to a series of 2-dim problems in
\cite{BKor_JIIPP,BKor_IP,BKor_SIAM}, where the results
corresponding to {\bf i\,-\,ii} are obtained. A characterization
relayed upon algebraic BCM is provided in \cite{BKor_CAOT}. Also,
levels {\bf i\,-\,iii} are reached in \cite{H&M} by the use of
multidimensional complex analysis.

Our paper deals with problem {\bf iv}. Its subject is a stability
of determination of the surface from its DN-map. Let the surfaces
$(M,g)$ and $(M',g')$ have the common boundary $\Gamma$. For
simplicity and with no lack of generality we assume that $g$ and
$g'$ induce on $\Gamma$ the same length element. Suppose that the
DN-map $\Lambda'$ is close to $\Lambda$. Can one claim that
$(M',g')$ is, in a certain sense, close to $(M,g)$? The rigorous
formulation of this question has to be prefaced with preliminary
discussion.

First of all, we need to provide a way of comparing the surfaces.
A variant is to consider diffeomorphisms $\kappa: \ M\to M'$ and
conformal factors obeying $\rho'|_{\Gamma}=1$, and compare the
pullback metric $g''=\kappa^{*}(\rho'g')$ with metric $g$. Then we
could say that $(M',g')$ is close to $(M,g)$ if there exist
$\kappa$ and $\rho'$ such that the metric $g''$ is close to $g$ on
$M$. Such a way looks natural but encounters the following
obstacle, which was rather unexpected for us: the surfaces of
different topology (of the Euler characteristics
$\chi(M)\not=\chi(M')$) can have arbitrarily close DN-maps.
Namely, the following is proved in \cite{ZNS}.
 \begin{proposition}
For any surface $(M,g)$ and any $\varepsilon>0$, $k,m=1,2,\dots$
there exists a surface $(M',g')$ such that
$\parallel\Lambda'-\Lambda\parallel_{H^{m}(\Gamma;\mathbb{R})\to
H^{m-1}(\Gamma;\mathbb{R})}<\varepsilon$ and $\chi(M')-\chi(M)=2k$
holds.
 \end{proposition}
In other words, topology of the surface is not stable with respect
to small perturbations of its DN-map. This motivates to impose the
additional condition
 \begin{equation}\label{Eq chi=chi'}
\chi(M')=\chi(M)
 \end{equation}
and we accept it for the rest of the paper. However, even under
the latter condition, there is no `canonical' diffeomorphism
$\kappa: M\to M'$ that enables to compare the surfaces and
metrics, and we propose another way of comparing.

\subsubsection*{Holomorphic immersions}
The idea is to compare not the surfaces $(M,g)$ and $(M',g')$
themselves, but their images $\mE(M)$ and $\mE'(M')$ in $\mathbb
C^n$ via the close holomorphic immersions $\mE$ and $\mE'$.
\smallskip

\noindent$\bullet$\,\,\, Recall the basic notions.

Let the surface $(M,g)$ be oriented with the smooth family of
`rotations' $\Phi=\{\Phi_x\}_{x\in M}$:
\begin{align*}
\Phi_x\in{\rm End\,}TM_x,\quad \Phi_x^*=\Phi_x^{-1}=-\Phi_x,
\end{align*}
that is equivalent to
\begin{align*}
& g(\Phi_{x}a,\Phi_{x}b)=g(a,b), \quad g(\Phi_{x}a,a)=0, \qquad
\forall a,b\in TM_x, \ x\in M.
\end{align*}
Then the boundary $\Gamma$ is oriented by the tangent field
$\gamma:=\Phi\nu$. In the sequel, dealing with a set of surfaces
$(M,g), (M',g'), (M'',g''), \dots$ with the common boundary
$(\Gamma,dl)$, we always assume that they are oriented in a
consistent way: $\Phi\nu=\Phi'\nu=\Phi''\nu=\dots=\gamma$.

A smooth function $w=\Re w+i\Im w:M\to\mathbb C$ is holomorphic if
the Cauchy-Riemann condition $\nabla\Im w=\Phi\nabla\Re w$ holds
in $M$. Its real and imaginary parts are harmonic: $\Delta_g\Re
w=\Delta_g\Im w=0$ holds in $M$. By $\mathfrak H(M)\subset
C(M;\mathbb C)$ we denote the lineal of holomorphic smooth
functions on $M$. Let $${\rm Tr}: C(M;\mathbb C)\to
C(\Gamma;\mathbb C),\quad h\mapsto h|_\Gamma$$ be the trace
operator. By the maximal principle, it maps the space $\fH(M)$
onto its image ${\rm Tr\,}\fH(M)$ isometrically.

\noindent$\bullet$\,\,\,Recall that an {\it immersion} is a
differentiable map $\kappa: M\mapsto M'$ between differentiable
manifolds $M$ and $M'$, whose differential $D\kappa: T_{x}M\mapsto
T_{\kappa(x)}M'$ is injective for all $x\in M_{1}$. We say that
the immersion
$$
\mE: M\to\mathbb C^n, \,\,\,x\mapsto \{w_{1}(x),\dots,w_{n}(x)\}
 $$
is {\it holomorphic} if it is realized by holomorphic functions $w_j$.
We deal with such immersions only and, for
short, call them just `immersions'. Since the reserve of harmonic
(and, hence, holomorphic) functions on $M$ is the same for the
metrics $g$ and $\rho g$, such immersion is determined by the
conformal class of the metric and the boundary values $\eta_{j}=w_{j}|_{\Gamma}$ of $w_{j}$. Hence, by \cite{LU,B}
it is determined by the DN-map $\Lambda$ and the choice of $\eta_{j}$.
\smallskip

\subsubsection*{Main result}
Recall the definition of the Hausdorff distance in $\mathbb C^n$.
Let $U^r[A]:=\{p\in\mathbb C^n \,|\,{\rm dist}_{\mathbb
C^n}(p,A)<r\}$ be the metric neighborhood of $A\subset\mathbb
C^n$. For the bounded sets $A,B\subset\mathbb C^n$ we set
 \begin{equation}\label{Eq def d Hauss}
{\rm dist}_H (A,B)\,:=\,\max\{r_{AB}, r_{BA}\},
 \end{equation}
where $r_{AB}:=\inf\{r>0\,|\,\,U^r[A]\supset B\}$ and
$r_{BA}:=\inf\{r>0\,|\,\,U^r[B]\supset A\}$.

Recall that we deal with the surfaces with the common boundary
$\Gamma$ and the convention (\ref{Eq chi=chi'}) is in force. By
$\mathbb{M}_{t}$ we denote the set of all such surfaces $(M',g')$,
whose the DN-maps $\Lambda'$ obey
\begin{equation}
\label{DNcloseness eq}
\parallel\Lambda'-\Lambda\parallel_{H^{1}(\Gamma;\mathbb{R})\to L_{2}(\Gamma;\mathbb{R})}\leqslant t.
\end{equation}
It is known (see \cite{LeeU}) that each DN-map is a first order
pseudo-differential operator. On the class of such operators, any
two norms $\parallel\cdot\parallel_{H^{s+1}(\Gamma;\mathbb{R})\to
H^{s}(\Gamma;\mathbb{R})}$ and
$\parallel\cdot\parallel_{H^{s'+1}(\Gamma;\mathbb{R})\to
H^{s'}(\Gamma;\mathbb{R})}$ with $s,s'\in\mathbb{R}$ are
equivalent. Thus, (\ref{DNcloseness eq}) is equivalent to the
following condition
\begin{equation}
\label{DNcloseness}
\parallel\Lambda'-\Lambda\parallel_{H^{3}(\Gamma;\mathbb{R})\to H^{2}(\Gamma;\mathbb{R})}\leqslant ct.
\end{equation}
\smallskip

The main result is the following
\begin{theorem}
\label{MT}
Let the holomorphic  immersion $\mE: M \to \mathbb C^n$ {\rm (}$n\geqslant 1${\rm)} be arbitrarily fixed. Then the relation
\begin{equation}
\label{supinf main}
\sup_{M'\in \mathbb{M}_{t}}\inf_{\mE'}d_{H}(\mE'(M'),\mE(M))\underset{t\to 0}{\longrightarrow}0,
\end{equation}
holds, where the infimum is taken over all holomorphic
immersions $\mE': M'\mapsto\mathbb C^n$.
\end{theorem}
The rest of the paper is devoted to the proof of Theorem \ref{MT}. First, we outline the sketch of the proof.

\noindent$\bullet$\,\,\, Let a surface $(M,g)$ and the immersion
$\mE:M\to\mathbb C^n$ be fixed. A surface $(M',g')$ with the same
boundary $(\Gamma=\Gamma', dl'=dl)$ is regarded as its
`perturbation'. To prove Theorem \ref{MT}, we construct, for each
$M'$, a certain map $\mE':M'\to\mathbb C^n,\,x\mapsto
\{w'_{1}(x),\dots,w'_{n}(x)\}$, which is determined by $M,\mE,M'$
and obeys
\begin{equation}\label{main result} \sup_{M'\in
\mathbb{M}_{t}}d_{H}(\mE'(M'),\mE(M))\underset{t\to
0}{\longrightarrow}0.
\end{equation}
The map $\mE'$ is connected with $\mE$ via the map $\beta'$:
\begin{equation}
\label{rule} w'_j=\beta'w_j,\,\,\,j=1,\dots, n,
\end{equation}
where $\beta': \fH(M)\to\fH(M')$ is a `canonical' real linear map
obeying
\begin{equation}\label{bcest}
\parallel{\rm Tr\,}'\beta'w-{\rm Tr\,}w\parallel_{C^{2}(\Gamma;\mathbb{C})}\leqslant c\,t\parallel{\rm Tr\,}w\parallel_{H^{3}(\Gamma;\mathbb{C})}, \qquad
w\in\fH(M);
\end{equation}
here and in the sequel, $c$'s are positive constants determined by
$M$ and $\mE$. As we show, for small enough $t$ and {\it all}
$M'\in\mathbb{M}_{t}$, the map $\mE'$ turns out to be an
immersion.

In what follows, we give a detailed description of the map
$\beta'$ along with the proof of (\ref{bcest}). In the subsequent,
we denote $\eta_{k}=w_{k}|_{\Gamma}$,
$\eta'_{k}=w'_{k}|_{\Gamma}$, so that $\mE$ and $\mE'$ are
determined by $\{\eta_{k}\}_{k=1}^{n}$ and
$\{\eta'_{k}\}_{k=1}^{n}$ by the uniqueness of analytic
continuation. Note that (\ref{bcest}) implies the closeness of the
boundary traces of $w'_{k}$ and $w_{k}$ for small $t$; namely, we
have
\begin{equation}\label{boundary curves estimate gen}
\parallel\eta'_{k}-\eta_{k}\parallel_{C^{2}(\Gamma;\mathbb{C})}\leqslant c\,t, \qquad k=1,\dots,n.
\end{equation}
In particular, the boundary images $\mE(\Gamma)$ and $\mE'(\Gamma)$ are close in $\mathbb C^n$ for small $t$.

\noindent$\bullet$\,\,\, Next, we prove that the estimate
(\ref{boundary curves estimate gen}) implies the closeness, by
Hausdorff distance, of the images $\mE'(M')$ and $\mE(M)$ for
small $t$. The key instrument of the proof is the {\it generalized
argument principle}. Suppose that ${w}$ and $\tilde w$ are
holomorphic smooth functions on $M$ and $z\not\in w(\Gamma)$ is
fulfilled. Then, from Stokes theorem (see Theorem 3.16, \cite{M})
and the residue theorem (see Lemma 3.12, \cite{M}) for meromorphic
1-form $(\tilde w/({w}-z))d{w}$, it follows that
\begin{equation}
\label{GAP}
\frac{1}{2\pi i}\int\limits_{\Gamma}\tilde
w\,\frac{\partial_{\gamma}{w}}{{w}-z}\,dl=\sum_{x\in {w}^{-1}(z)}
m({w},x,z)\,\tilde w(x)
\end{equation}
holds, where $m({w},x,z)$ is the order of the zero $x$ of
the function $w-z$. If $w-z$ has a unique zero $x$ on $M$ and its order is one, then the
formula above is simplified as follows
\begin{equation}
\label{AP}
\frac{1}{2\pi i}\int\limits_{\Gamma}\tilde
w\,\frac{\partial_{\gamma}{w}}{{w}-z}\,dl=\tilde w(x),
\end{equation}
so we can find the values at $x$ of all holomorphic functions
$\tilde{w}$. Of course, the same facts is true for $M'$ instead of
$M$. So, we can try to take one of the projections $w_{j}(x)$
($w'_{j}(x')$) as a coordinate of the point $p=\mE(x)$
($p'=\mE'(x')$) and determine all other projections $w_{k}(x)$
($w'_{k}(x')$) via formula (\ref{AP}). Thereby, we determine the
images $\mE(M)$ and $\mE'(M')$. The class of {\it projective}
immersions, for which it is possible, is described below. Also, it
is shown that it suffices to prove (\ref{main result}) only for
such projective immersions.

Due to (\ref{AP}), the closeness of $\eta'_{k}$ and $\eta_{k}$ implies the closeness of
the surface images $\mE(M)$ and $\mE'(M')$ in $\mathbb C^n$ outside a small neighborhood of $\mE(\Gamma)$. Near $\mE(\Gamma)$,
the application of the generalized argument
principle is reduced to the standard estimates of the Cauchy-type
integrals by choosing appropriate local coordinates for points
$\mE(M)$ and $\mE'(M')$. Summarizing, we will arrive at (\ref{main result}).

\section{The map $\beta'$}
\subsubsection*{Preliminaries}
\noindent$\bullet$\,\,\,The following is a few known facts. Introduce the sets of functions with the zero mean value
\begin{align*}
\dot L_2&(\Gamma,\mathbb R):=\{f\in L_2(\Gamma,\mathbb
R)\,|\,\,\langle f\rangle=0\}, \quad \langle f\rangle:=\int_\Gamma
f\,dl,\\
\dot C^\infty&(\Gamma;\mathbb R):= \dot L_2(\Gamma;\mathbb R)\cap
C^\infty(\Gamma;\mathbb R),\\
\dot H^m&(\Gamma;\mathbb R):= \dot L_2(\Gamma;\mathbb R)\cap
H^m(\Gamma;\mathbb R),
\end{align*}
where $H^m(...)$ are the Sobolev spaces. As is known, the DN map
$\Lambda$ is well defined on smooth functions and acts
continuously from $H^m(\Gamma,\mathbb R)$ to
$H^{m-1}(\Gamma,\mathbb R),\,\,\,m=1,2,\dots$. It preserves the
smoothness, its (operator) closure is a self-adjoint first-order
pseudo-differential operator in $L_2(\Gamma;\mathbb R)$, and
 \begin{equation}\label{Eq Lambda properties}
{\rm Ker\,}\Lambda=\{\rm const\},\quad\,\,\overline{{\rm
Ran\,}\Lambda}=\dot L_2(\Gamma,\mathbb R)
 \end{equation}
holds.

Let $\partial_\gamma$ be the differentiation on the boundary with
respect to the field $\gamma=\Phi\nu$. The integration
$J=\partial_\gamma^{-1}$ is well defined in $\dot
L_2(\Gamma,\mathbb R)$, whereas the operator $J\Lambda$ is also
well defined due to (\ref{Eq Lambda properties}) and is bounded in
each $H^m(\Gamma;\mathbb R),\,\,m=0,1,\dots$.
Note that if $M$ is a disk in
$\mathbb R^2$, the operator $J\Lambda$ coincides with the
classical Hilbert transform that relates the real and imaginary
parts of the trace of holomorphic function. In the general case,
for $w\in\fH(M)$ and $\eta:=w|_\Gamma={\rm Tr\,}w$ we also have
 \begin{equation}\label{Eq eta unperturbed}
\eta = \Re \eta+i\,[J\Lambda \Re \eta+\langle\Im \eta\rangle]
 \end{equation}
(see \cite{B}).

\noindent$\bullet$\,\,\,By (\ref{Eq Lambda properties}), $\Lambda
J$ is a well defined on $\dot C^\infty(\Gamma;\mathbb R)$ and
bounded operator acting in $\dot L_2(\Gamma;\mathbb R)$. It
preserves smoothness: $\Lambda J\,\dot H^m(\Gamma;\mathbb R)\subset
\dot H^m(\Gamma;\mathbb R)$ holds for all $m=0,1,\dots$. The
representation $(\Lambda J)^*=-J\Lambda=J(\Lambda
J)\partial_\gamma$ is valid on smooth functions and implies
 \begin{equation}\label{Eq L1 1}
[\,\mathbb I+\left(\Lambda J\right)^2]^*=\mathbb I+\left((\Lambda
J)^*\right)^2=\mathbb I+J\left(\Lambda
J\right)^2\partial_\gamma=J\,[\,\mathbb I+\left(\Lambda
J\right)^2]\partial_\gamma\qquad {\rm on}\,\,\,\dot
C^\infty(\Gamma;\mathbb R),
 \end{equation}
where $\mathbb I=\partial_\gamma J$ is the unit operator in $\dot
L_2(\Gamma;\mathbb R)$.

As is shown in \cite{B}, Lemma 1, the relations
\begin{align*}
{\rm Ran\,}[\,\mathbb I+\left(\Lambda
J\right)^2]\subset \dot C^\infty(\Gamma;\mathbb R),\\
{\rm dim\,Ran\,}[\,\mathbb I+\left(\Lambda
J\right)^2]=1-\chi(M)=:\varkappa
\end{align*}
hold; hence, the operator $\mathbb I+\left(\Lambda
J\right)^2$ determines the topology of $M$. By (\ref{Eq L1 1}), these facts
easily lead to
 \begin{equation*}
{\rm dim\,Ran\,}[\,\mathbb I+\left(\Lambda
J\right)^2]^*=\varkappa,
 \end{equation*}
so that $[\,\mathbb I+\left(\Lambda J\right)^2]^*$ is a finite
rank operator in $\dot L_2(\Gamma;\mathbb R)$. The relations
 \begin{equation}\label{Eq L1 3}
 \begin{split}
\overline{\Re {\rm Tr\,}\fH(M)}={\rm Ker\,}[\,\mathbb
I+\left(\Lambda J\right)^2]\oplus \mathbb R,\\
\overline{\Re {\rm Tr\,}\fH(M)}\cap C^\infty(\Gamma;\mathbb R)=\Re {\rm Tr\,}\fH(M).
\end{split}
 \end{equation}
are also proved in \cite{B}. As a consequence, we have
 \begin{align}
\notag & L_2(\Gamma;\mathbb R)=\dot L_2(\Gamma;\mathbb
R)\oplus\mathbb R={\rm Ker\,}[\,\mathbb I+\left(\Lambda
J\right)^2]\oplus{\rm Ran\,}[\,\mathbb I+\left(\Lambda
J\right)^2]^*\oplus\mathbb R\overset{(\ref{Eq L1 3})}=\\
\label{Eq L1 4} &= \overline{\Re {\rm Tr\,}\fH(M)}\oplus{\rm
Ran\,}[\,\mathbb I+\left(\Lambda J\right)^2]^*.
\end{align}
Let $P$ be the projection in $L_2(\Gamma;\mathbb R)$ onto the
subspace $\overline{\Re {\rm Tr\,}\fH(M)}$. Then, for any $\eta\in C^\infty(\Gamma;\mathbb R)$,
we have $P\eta\in\Re {\rm Tr\,}\fH(M)$. Of course, the same facts are true for the surface $M'$ instead of $M$.

\subsubsection*{Operators $\beta'_{\Gamma}$ and $\beta'$}
Let $P'$ be the projection in $L_2(\Gamma;\mathbb R)$ onto the
subspace $\overline{\Re {\rm Tr\,}\fH(M')}$. The map
 \begin{equation}\label{Eq eta perturbed}
\beta'_{\Gamma}\eta:=P'\Re\eta+i\,[J\Lambda'P'\Re\eta+\langle\Im\eta\rangle],
\qquad \eta\in C^{\infty}(\Gamma;\mathbb{C})    
 \end{equation}
provides a natural way to relate the traces of $\fH(M)$ with the
traces of $\fH(M')$. In the meantime, the map
 \begin{equation*}
\beta'\,:=\,({\rm Tr\,}')^{-1}\beta'_\Gamma\,{\rm Tr\,},
 \end{equation*}
where ${\rm Tr\,}'$ is the trace operator on $\fH(M')$, is well
defined and relates $\fH(M)$ with $\fH(M')$. Now, we prove estimate (\ref{bcest}).
\begin{lemma}\label{L1}
For sufficiently small $t>0$ and for any $M'\in\mathbb{M}_{t}$,
the map $\beta'_\Gamma$ is a bijection and the inequality
\begin{equation}\label{boundary curves estimate V}
\parallel\beta'_{\Gamma}\eta-\eta\parallel_{C^{2}(\Gamma;\mathbb{C})}\leqslant c\,t\parallel\eta\parallel_{H^{3}(\Gamma;\mathbb{C})}, \qquad \eta\in {\rm Tr\,}\fH(M)
\end{equation}
holds with a constant $c$ independent of $M'$.
\end{lemma}
\begin{proof}
Let $Q$ be the (finite rank) projection in $L_2(\Gamma;\mathbb R)$
onto ${\rm Ran\,}[\,\mathbb I+\left(\Lambda
J\right)^2]^*\overset{(\ref{Eq L1 1})}=J\,[\,\mathbb
I+\left(\Lambda J\right)^2]\partial_\gamma C^{\infty}(\Gamma;\mathbb R)$. Choose the smooth
$f_1,\dots,f_\varkappa$ so that $h_j:=J\,[\,\mathbb
I+\left(\Lambda J\right)^2]\partial_\gamma
f_j,\,\,\,j=1,\dots,\varkappa$ constitute a basis in ${\rm
Ran\,}[\,\mathbb I+\left(\Lambda J\right)^2]^*={\rm Ran\,}Q$. Repeat the same construction for the projection $Q'$ in
$L_2(\Gamma;\mathbb R)$ onto ${\rm Ran\,}[\,\mathbb
I+\left(\Lambda' J\right)^2]^*=J\,[\,\mathbb
I+\left(\Lambda' J\right)^2]\partial_\gamma C^{\infty}(\Gamma;\mathbb R)$, and put
$h'_j:=J\,[\,\mathbb I+\left(\Lambda' J\right)^2]\partial_\gamma
f_j,\,\,\,j=1,\dots,\varkappa$ (with the same $f_j$). Representing
 \begin{align*}
& h'_j-h_j=\left\{J\,[\,\mathbb I+\left(\Lambda'
J\right)^2]\partial_\gamma-J\,[\,\mathbb I+\left(\Lambda
J\right)^2]\partial_\gamma\right\}f_j=\\
& = J\left[\Lambda'J(\Lambda'J-\Lambda J)+(\Lambda'J-\Lambda
J)\Lambda
J\right]\,\partial_\gamma\,f_j=J\left[(\Lambda'J)(\Lambda'-\Lambda)+(\Lambda'-\Lambda
)(J\Lambda)\right]\,f_j
 \end{align*}
and taking into account the boundedness of $J, \Lambda'J, \Lambda
J$ as operators in $H^m(\Gamma;\mathbb R)$ and $\Lambda', \Lambda$
as operators from $H^m(\Gamma;\mathbb R)$ to
$H^{m-1}(\Gamma;\mathbb R)$, one easily gets
 $$
\| h'_j-h_j\|_{H^m(\Gamma;\mathbb R)}\leqslant
c_m\,\|\Lambda'-\Lambda\|_{H^m(\Gamma;\mathbb R)\to
H^{m-1}(\Gamma;\mathbb R)}\,\|f_j\|_{H^m(\Gamma;\mathbb R)}.
 $$
Hence, by virtue of (\ref{DNcloseness}), the inequality
 \begin{equation}\label{Eq L1 4}
\|h'_j-h_j\|_{H^3(\Gamma;\mathbb R)}\leqslant
c\,t\,\|f_j\|_{H^3(\Gamma;\mathbb R)},\qquad j=1,\dots,\varkappa
 \end{equation}
holds with a constant $c$ independent on $M'\in\mathbb{M}_{t}$.

By the latter estimate and assumption (\ref{Eq chi=chi'}), for  small enough $t$, the system
$h'_1,\dots,h'_\varkappa$ forms a basis in its linear span that is
${\rm Ran\,}[\,\mathbb I+\left(\Lambda' J\right)^2]^*={\rm
Ran\,}Q'$. Closeness of the bases $h'_1,\dots,h'_\varkappa$ and
$h_1,\dots,h_\varkappa$ leads to the closeness of the projections:
by the use of (\ref{Eq L1 4}) it is easy to verify that the
inequality
 \begin{equation}\label{Eq L1 5}
\|Q'-Q\|_{H^3(\Gamma;\mathbb R)\to H^3(\Gamma;\mathbb R)}\leqslant
c\,t
 \end{equation}
holds. Comparing (\ref{Eq L1 3}) and (\ref{Eq L1 4}), we obtain $P=\mathbb I-Q$ and $P'=\mathbb I-Q'$.
Then (\ref{Eq L1 5}) is equivalent to
 \begin{equation}\label{Eq L1 5 P}
\|P'-P\|_{H^3(\Gamma;\mathbb R)\to H^3(\Gamma;\mathbb
R)}\leqslant {\rm const}\,t.
 \end{equation}
Fix a smooth $\eta\in{\rm Tr\,}\Re\fH(M)$; then $\Re\eta=P\Re\eta$. Comparing (\ref{Eq eta
unperturbed}) with (\ref{Eq eta perturbed}), we have
 \begin{align*}
&
\beta'_\Gamma\eta-\eta=(P'-P)\,\Re\eta+i\,\left[J(\Lambda'P'-\Lambda
P)\right]\,\Re\eta\,=\\
& =
(P'-P)\,\Re\eta+i\,\left[(J\Lambda')(P'-P)+J(\Lambda'-\Lambda)P\right]\,\Re\eta.
 \end{align*}
From (\ref{Eq L1 5 P}) and (\ref{DNcloseness}), we easily obtain that
$$
\|\beta'_\Gamma\eta-\eta\|_{H^3(\Gamma;\mathbb R)}\leqslant
c\,t\,\|\eta\|_{H^3(\Gamma;\mathbb R)}
 $$
holds with a constant $c$ independent on $M'\in\mathbb{M}_{t}$. At last, by
continuity of the embedding $H^3(\Gamma;\mathbb R)\subset
C^2(\Gamma;\mathbb R)$, we arrive at (\ref{boundary curves
estimate V}). By (\ref{boundary curves
estimate V}), for small $t$, $\beta'_{\Gamma}$ is a (close to identity)
invertible map from ${\rm Tr\,}\fH(M)$ to ${\rm Tr\,}\fH(M')$,
which preserves the smoothness. As a result, $\beta'_{\Gamma}$
is an invertible map from $\fH(M)$ to $\fH(M')$. 
\end{proof}
As a corollary of (\ref{boundary curves estimate V}), we obtain (\ref{boundary curves estimate gen}) and
\begin{align}\label{Eq Close bound images}
\|\mE'(x)-\mE(x)\|_{\mathbb C^n}\leqslant c\,t,\qquad x\in\Gamma,
\end{align}
where $\mE'$ is the map given by (\ref{rule}). So, the smallness
of $\Lambda'-\Lambda$ yields the closeness of the images of the
(common) boundary $\Gamma$ via the holomorphic maps $\mE$
and $\mE'$. It is the fact that motivates the further use of the
pair $\mE,\mE'$ for evaluation of the closeness of the surfaces
$M$ and $M'$.

\section{Closeness outside the boundary}
\label{Sec2}

We now proceed to the proof of (\ref{main result}). In
accordance with the plan outlined in Introduction, we establish
the closeness of the compact set $K\subset \mE(M)$ and 
a suitable $K'\subset \mE(M')$, which are located outside
$\mE(\Gamma)$ and $\mE(\Gamma')$ respectively.

\subsubsection*{Projective immersions}

$\bullet \ $ First, we show that it is sufficient to prove
(\ref{main result}) only for $\mE$ belonging to a special class of
immersions. Such a class is defined in the following way.

Let $\xi=\{\xi_{1},\dots,\xi_{n}\}\in \mathbb C^n$; by
$\pi_j:\xi\mapsto \xi_j$ we denote the coordinate projection. Let
$B_{r,j}[\xi]:=\{\zeta\in \mathbb C^n \,|\,
|\zeta_{j}-\xi_{j}|<r\}$ be a cylinder in $\mathbb C^n$.
\begin{definition}
\label{project emb} The immersion $\mE: \ M\to\mathbb{C}^{n}$ is
called projective if for each point $\xi\in \mE(M)$ there is a
cylinder $B_{r,j}[\xi]$ such that the projection $\pi_{j}: \
\mE(M)\cap B_{r,j}[\xi]\to \mathbb{C}$ is a diffeomorphism.
\end{definition}
Of course, the projective immersion is an embedding. Also,
Definition \ref{project emb} provides the following property of
projective immersions: if $z\in\pi_{j}(\mE(M)\cap B_{r,j}[\xi])$,
then the coordinate plane $\{\zeta\in\mathbb{C}^{n} \ | \
\zeta_{j}=z\}$ crosses $\mE(M)$ at a single point.

\begin{proposition}
The projective immersions do exist.
 \end{proposition}
\begin{proof}
Due to smoothness of the boundary $\Gamma$ the surface $M$ can be
embedded into a larger noncompact surface $\tilde{M}$. In view of
divisor theorem (see Proposition 26.5, \cite{Forster}), for any
$x\in M$ there exists a holomorphic function $\tilde{w}_{x}$ on
$\tilde{M}$ such that $x$ is a single zero of $\tilde{w}_{x}$ and
the differential $d\tilde{w}_{x}$ at $x$ is injective. Then there
exists a neighborhood $U_{x}$ of $x$ in ${M}$ such that
$\tilde{w}_{x}: \ U_{x}\to \mathbb{C}$ is an injection and the
differential $d\tilde{w}_{x}$ at any point of $U_{x}$ is
injective. The set $M\backslash U_{x}$ is compact, so that
$|\tilde{w}_{x}|\ge r_{x}>0$ holds on $M\backslash U_{x}$.

Let $\tilde U_{x}:=\{y\in U_{x} \ | \ |\tilde{w}_{x}|<r_{x}/2\}$.
The neighbourhoods $\tilde U_{x}$, $x\in M$ constitute an open
cover of $M$, from which one can choose a finite subcover $\tilde
U_{x_{j}}$, $j=1,\dots,n$. Taking $\mE: M\to\mathbb
C^n,\,\,\mE=\{w_1(\cdot),\dots,w_n(\cdot)\}$, $w_{j}=\tilde
w_{x_{j}}|_{M}$, we obtain the immersion, which is projective by
construction.
\end{proof}
Let's note a simple fact. If $\mE_{1}: \ M\to\mathbb{C}^{n_{1}}$
and $\mE_{2}: \ M\to\mathbb{C}^{n_{2}}$ are immersions, whereas
$\mE_{2}$ is projective, then the immersion
 $$
\mE_{1}\oplus\mE_{2}: \ M\to\mathbb{C}^{n_{1}+n_{2}}, \qquad
\left(\mE_{1}\oplus\mE_{2}\right)(x):= \{\mE_{1}(x),\mE_{2}(x)\}
 $$
is also projective. It is used as follows.
\smallskip

\noindent$\bullet \ $ Proving estimate (\ref{main result}), we can
restrict ourselves only to the case of projective immersions.
Indeed, assume that (\ref{main result}) is already proved for any
$n=1,2,\dots,$ and any projective immersion $\mE_{0}$. Let $\mE:\
M\to \mathbb{C}^{m}$ be an immersion of $M$ and let $M'\in
\mathbb{M}_{t}$ be chosen arbitrarily. Take some projective
immersion $\mE_{0}: \ M\to \mathbb{C}^{s}$ of $M$. Then $\mE\oplus
\mE_{0}$ is also a projective and, by our assumption,
 $$
\sup_{M'\in \mathbb{M}_{t}}d_{H}([(\mE\oplus \mE_{0})'](M'),[\mE\oplus \mE_{0}](M))\underset{t\to
0}{\longrightarrow}0
 $$
holds. Also, we have
 $$
(\mE\oplus \mE_{0})'=\mE'\oplus \mE_{0}',\qquad
\mathcal\varpi^{n}_{m}\circ[\mE\oplus \mE_{0}]=\mE, \qquad
\mathcal\varpi^{n}_{m}\circ[\mE'\oplus \mE_{0}']=\mE',
$$
where $n=m+s$ and the map $\varpi^{n}_{m}$,
\begin{equation}
\label{projection le s}
\varpi^{n}_{m}\{\xi_{1},\dots,\xi_{n}\}:=\{\xi_{1},\dots,\xi_{m}\}
\end{equation}
projects on the first component of the sum $\mathbb
C^{m}\oplus\mathbb C^{s}=\mathbb C^{n}$. In view of the evident
relations
 $$
d_{H}(\mE'(M'),\mE(M))=d_{H}(\varpi^{n}_{m}\circ[\mE'\oplus
\mE_{0}'](M'),\varpi^{n}_{m}\circ[\mE\oplus \mE_{0}](M))\leqslant
d_{H}([\mE'\oplus \mE_{0}'](M'),[\mE\oplus \mE_{0}](M)),
 $$
we arrive at $\sup_{M'\in
\mathbb{M}_{t}}d_{H}(\mE'(M'),\mE(M))\underset{t\to
0}{\longrightarrow}0.$

Thus, if estimate (\ref{main result}) holds for projective
immersions then it is valid for all immersions
$\mE:M\to\mathbb C^n$. For this reason, in the sequel we assume
that
\smallskip

\centerline{all the maps $\mE$, which we deal with, are
projective.} \bigskip

\noindent$\bullet \ $ Suppose that $\mE:M\to\mathbb C^n,
\mE=\{w_1(\cdot),\dots,w_n(\cdot)\}$ is a (projective) immersion
and $\mE': M'\to\mathbb{C}^{n}$ is the corresponding map given by
(\ref{rule}). We first estimate the closeness, in Hausdorff
distance, between the parts of $\mE(M)$ and $\mE'(M)$ which are
separated from the curves $\mE(\Gamma)$ and $\mE'(\Gamma)$
respectively.

Denote by $V_{j}^\varepsilon$ the set of all $z\in\mathbb{C}$ such
that ${\rm dist}_{\mathbb{C}}(z,w_{j}(\Gamma))>\varepsilon$ and
the function $w_{j}-z$ has a simple (of the order 1) zero in $M$
and has no more zeroes in $M$. Some of $V_{j}^\varepsilon$ may be
empty; however, to deal with them is reasonable, when
immersion $\mE(M)$ is projective: see below.

Introduce the `cylinders'
 $$
\Pi_{j}^\varepsilon:=\underset{1}
{\mathbb{C}}\times\dots\times\underset{j}
{V_{j}^\varepsilon}\times\dots\times\underset{n} {\mathbb{C}}
 $$
and consider their union
$$\mathtt{Q}_{\varepsilon}\,:=\,\bigcup_{j=1}^{n}\Pi_{j}^\varepsilon\,\subset\,\mathbb C^n.$$
If a subset $K\subset\mE(M\backslash\Gamma)$ is compact then there
exists $\varepsilon_{0}>0$ such that
$K\subset\mathtt{Q}_{\varepsilon}$ for any $\varepsilon\in
(0,\varepsilon_{0})$. Indeed, since $\mE$ is projective, for any
$\xi\in K$ one can specify the number $j(\xi)$ and an $s(\xi)>0$
such that $\pi_{j(\xi)}\,\xi\in V_{j(\xi)}^{s(\xi)}$ holds.
Therefore, we have $\xi\in\Pi_{j(\xi)}^{s(\xi)}$. The sets
$\Pi_{j(\xi)}^{s(\xi)}\cap K$, $\xi\in K$ constitute an open cover
of $K$, from which one can choose a finite subcover
$\Pi_{j(\xi_{k})}^{s(\xi_{k})}\cap K$, $k=1,\dots,N$. Denote by
$\varepsilon_{0}=\underset{k}\min\{s(\xi_{k})\}$; then
$K\subset\mathtt{Q}_{\varepsilon}$ does hold for any
$\varepsilon\in (0,\varepsilon_{0})$.
\smallskip

\subsubsection*{Estimates}

$\bullet \ $ Now, we prove that the parts of the images $\mE(M)$
and $\mE'(M)$ that are contained in $\mathtt Q^\varepsilon$, are
close (if $\Lambda'$ is close to $\Lambda$). The proof is based on
the application of the generalized argument principle (\ref{GAP}).
The bar denotes the closure in $\mathbb C^n$; the projection
$\varpi^{n}_{m}$ is given by (\ref{projection le s}).
 \begin{lemma}
 \label{L2}
\

{\rm i}. The relation
 \begin{equation}
\label{main result on compacts} \sup_{M'\in
\mathbb{M}_{t}}d_{H}(\mE'(M')\cap\overline{\mathtt{Q}_{\varepsilon}},\mE(M)\cap\overline{\mathtt{Q}_{\varepsilon}})=O(t),
\qquad t\to 0
\end{equation}
holds for any fixed $\varepsilon>0$.

{\rm ii}. Suppose that the map $\varpi^{n}_{m}\mE$, $m<n$ is an
immersion. 
Then, for sufficiently small $t>0$ and any $M'\in \mathbb{M}_{t}$,
the map $\varpi^{n}_{m}\mE': \
\mE^{'-1}(\mathtt{Q}_{\varepsilon})\mapsto \mathbb{C}^{m}$ is
also an immersion.
 \end{lemma}
\begin{proof}
\noindent$\bf 1.$\,\,\,
Let $\varepsilon>0$ be fixed and $\xi$ be
an arbitrary point of
$\mE(M)\cap\overline{\mathtt{Q}_{\varepsilon}}$. Then $\xi\in
\overline{\Pi_{j}^\varepsilon}$ for some $j$. Put
$z:=\pi_{j}\,\xi$, then $z\in \overline{V_{j}^\varepsilon}$ holds.
Also, put by definition $w_{0}=1$ and $w_{0}'=1$ on $M$ and $M'$
respectively. Recall that $\mE=\{w_1(\cdot),\dots,w_n(\cdot)\}$,
$\mE'=\{w'_1(\cdot),\dots,w'_n(\cdot)\}$ and $\eta_j=w_j|_\Gamma$,
$\eta'_j=w'_j|_\Gamma$. Denote
\begin{equation*}
J_{k,j}(z):=\frac{1}{2\pi
i}\int_{\Gamma}\eta_{k}\,\frac{\partial_{\gamma}\eta_{j}}{\eta_{j}-z}\,dl,
\qquad J'_{k,j}(z):=\frac{1}{2\pi
i}\int_{\Gamma}\eta'_{k}\,\frac{\partial_{\gamma}\eta'_{j}}{\eta'_{j}-z}\,dl.
\end{equation*}
By the generalized argument principle (\ref{GAP}), we have
\begin{equation}
\label{argument principle application} J_{k,j}(z)=\sum_{x\in
w_{j}^{-1}(z)} w_{k}\,(x)\,m(w_{j},x,z), \qquad
J'_{k,j}(z)=\sum_{x'\in {w'_{j}}^{-1}(z)}
w'_{k}(x')\,m'(w'_{j},x',z),
\end{equation}
where $k=0,\dots,n$ and the number $m(w_{j},x,z)$ (the number
$m'_{j}(w'_{j},x',z)$) is the order of zero $x$ (of zero $x'$) of
the function $w_{j}-z$ (the function $w'_{j}-z$) on the surface
$M$ (on $M'$). In particular, the number $J_{0,j}$  is the
total multiplicity of zeroes of the function $w_{j}-z$ on $M$,
whereas $J_{0,j}=1$ by the choice of the coordinate functions of
the projective immersion $\mE$. In the meantime, $J'_{0,j}$ is the
total multiplicity of zeroes of the function $w'_{j}-z$ on $M'$.
\smallskip

\noindent$\bf 2.$\,\,\,Let us estimate $J'_{k,j}(z)-J_{k,j}(z)$. We have
\begin{align*}
& |J'_{k,j}(z)-J_{k,j}(z)|=\Big|\frac{1}{2\pi i}\int_{\Gamma}\Big(\eta'_{k}\,\frac{\partial_{\gamma}\,\eta'_{j}}{\eta'_{j}-z}-\eta_{k}\,\frac{\partial_{\gamma}\,\eta_{j}}{\eta_{j}-z}\Big)\,dl\,\Big|=\\
& =\Big|\frac{1}{2\pi
i}\int_{\Gamma}\frac{(\eta_{j}-z)\,\eta'_{k}\,\partial_{\gamma}\eta'_{j}-(\eta'_{j}-z)\,\eta_{k}\,\partial_{\gamma}\eta_{j}}{(\eta'_{j}-z)(\eta_{j}-z)}\,dl\,\Big|.
\end{align*}
Also,
\begin{align*}
& |\partial_{z}J'_{k,j}(z)-\partial_{z}J_{k,j}(z)|=\Big|\frac{1}{2\pi i}\int_{\Gamma}\Big(\eta'_{k}\,\frac{\partial_{\gamma}\,\eta'_{j}}{(\eta'_{j}-z)^{2}}-\eta_{k}\,\frac{\partial_{\gamma}\,\eta_{j}}{(\eta_{j}-z)^{2}}\Big)\,dl\,\Big|.
\end{align*}
From the inclusion $z\in \overline{V_{j}^{\,\varepsilon}}\subset
V_{j}^{\,\varepsilon/2}$ it follows that $J_{0,j}=1$ and ${\rm
dist}_{\mathbb{C}}(z,w_{j}(\Gamma))>\varepsilon$. By the latter
inequality and (\ref{boundary curves estimate gen}), we have
\begin{equation}
\label{nonsmalldenom}
|\eta_{j}-z|>\varepsilon, \qquad |\eta'_{j}-z|\geqslant
|\eta_{j}-z|-\parallel
\eta'_{j}-\eta_{j}\parallel_{C(\Gamma;\mathbb{C})}\geqslant\varepsilon-O(t)\geqslant
\varepsilon/2
\end{equation}
on $\Gamma$ for any sufficiently small $t>0$. Therefore,
\begin{align*}
& |J'_{k,j}(z)-J_{k,j}(z)|\leqslant
c(\varepsilon)\parallel(\eta_{j}-z)\,\eta'_{k}\,\partial_{\gamma}\eta'_{j}-
(\eta'_{j}-z)\,\eta_{k}\,\partial_{\gamma}\eta_{j}\parallel_{C(\Gamma;\mathbb{C})}=\\
&
=c(\varepsilon)\parallel(\eta_{j}-z)\,(\eta'_{k}-\eta_{k})\,\partial_{\gamma}\eta'_{j}+
(\eta_{j}-z)\,\eta_{k}\,(\partial_{\gamma}\eta'_{j}-\partial_{\gamma}\eta_{j})+(\eta_{j}-\eta'_{j})\,\eta_{k}\,\partial_{\gamma}\eta_{j}\parallel_{C(\Gamma;\mathbb{C})}\leqslant \\
&\leqslant c(\varepsilon)\big(\parallel\eta_{j}-z\parallel_{C(\Gamma;\mathbb{C})}\parallel\eta'_{j}\parallel_{C(\Gamma;\mathbb{C})}\parallel\eta'_{j}\parallel_{C^{1}(\Gamma;\mathbb{C})}+\\
& +\parallel\eta_{j}-z\parallel_{C(\Gamma;\mathbb{C})}\parallel\eta_{l}\parallel_{C(\Gamma;\mathbb{C})}\parallel\eta'_{j}-\eta_{j}\parallel_{C^{1}(\Gamma;\mathbb{C})}+\\
&
+\parallel\eta_{j}-\eta'_{j}\parallel_{C(\Gamma;\mathbb{C})}\parallel\eta_{l}\parallel_{C(\Gamma;\mathbb{C})}\parallel\eta_{j}\parallel_{C^{1}(\Gamma;\mathbb{C})}\big)
\end{align*}
holds with $c(\varepsilon)$ independent of $x\in
\mE(M)\cap\overline{\mathtt{Q}_{\varepsilon}}$. Now,
(\ref{boundary curves estimate gen}) yields the estimate
\begin{equation}
\label{boundary integral estimates}
|J'_{k,j}(z)-J_{k,j}(z)|=O(t).
\qquad t\to 0,
\end{equation}
Similarly, (\ref{nonsmalldenom}) and (\ref{boundary curves estimate gen})
imply
\begin{equation}
\label{boundary integral estimates - diff}
|\partial_{z}J'_{k,j}(z)-\partial_{z}J_{k,j}(z)|=O(t).
\qquad t\to 0,
\end{equation}
Note that estimates (\ref{boundary integral estimates}), (\ref{boundary integral estimates - diff})
are uniform with respect to $z\in V_{j}^\varepsilon$ and $M'\in
\mathbb{M}_{t}$.
\smallskip

\noindent$\bf 3.$\,\,\, Let $x\in M$ be the (single,
simple) zero of the function $w_j-z$, which corresponds to the
equality $J_{0,j}=1$. Since $J'_{0,j}$ is a total multiplicity of
zeroes of $w'_{j}-z$ on $M'$ and $J_{0,j}=1$, (\ref{boundary
integral estimates}) implies that, for all sufficiently small
$t>0$, the function $w'_{j}-z$ has a single zero $x'=x'(x)$ on
$M'$ and this zero is simple. At this point, we establish
a bijection $x\leftrightarrow x'$ between some points of $M$ and
$M'$. It induces the bijection
 \begin{equation}\label{Eq First bijection}
\mE(x)=\,\xi\longleftrightarrow \,\xi'=\mE'(x')
 \end{equation}
between the parts
$\mE'(M')\cap\overline{\mathtt{Q}_{\varepsilon}}$ and
$\mE(M)\cap\overline{\mathtt{Q}_{\varepsilon}}$ of the surface
images. The latter makes it possible to estimate the distance
between the parts as follows.

By (\ref{argument principle application}) and (\ref{boundary
integral estimates}) we have
 $$
|\pi_{k}\mE'(x')-\pi_{k}\mE(x)|=|w'_{k}(x')-w_{k}(x)|=O(t), \qquad t\to 0;\quad k=1,\dots, n.
 $$
Hence we obtain the estimate
 $$
|\mE'(x')-\mE(x)|=O(t), \qquad t\to 0,
 $$
which is uniform with respect to $\xi\in
\mE(M)\cap\overline{\mathtt{Q}_{\varepsilon}}$ and $M'\in
\mathbb{M}_{t}$. This implies (\ref{main result on compacts}).
\smallskip

\noindent$\bf 4.$\,\,\, Also, estimates (\ref{argument principle application}) and (\ref{boundary
integral estimates - diff}) imply
\begin{equation}
\label{est diff}
|\partial_{z}w'_{k}(x')-\partial_{z}w_{k}(x)|=O(t), \qquad t\to 0;\quad k=1,\dots, n,
\end{equation}
where $z:=\pi_{j}\mE(x)=\pi_{j}\mE'(x')\in V_{j}^\varepsilon$. Let $m<n$. Taking $\Re z$, $\Im z$ as local coordinates of the point $x\in M$ ($x'\in M'$), consider the maps
\begin{align*}
\hat{\mE}: \ (\Re z,\Im z)\mapsto (\Re w_{1},\dots,\Re w_{m},\Im w_{1},\dots,\Im w_{m}),\\
\hat{\mE'}: \ (\Re z,\Im z)\mapsto (\Re w'_{1},\dots,\Re w'_{m},\Im w'_{1},\dots,\Im w'_{m})
\end{align*}
Denote by ${\rm Jac\,}\hat{\mE}$ and ${\rm Jac\,}\hat{\mE}'$ the
Jacobian matrices of $\hat{\mE}$ and $\hat{\mE}'$, respectively.
Then (\ref{est diff}) yields
\begin{equation}
\label{est diff 1}
\parallel {\rm Jac\,}\hat{\mE}'(z)-{\rm Jac\,}\hat{\mE}(z)\parallel=O(t), \qquad t\to 0;\quad k=1,\dots, n.
\end{equation}
Suppose that the map $\varpi^{n}_{m}\mE$ is an immersion, then,
for any $z\in V_{j}^\varepsilon$, the matrix ${\rm
Jac\,}\hat{\mE}(z)$ is of full rank. In view of (\ref{est
diff 1}), there exist a sufficiently small $t_{j}>0$, such that
the matrix ${\rm Jac\,}\hat{\mE}'(z)$ is of full rank for any
$t\in (0,t_{j})$, $M'\in\mathbb{M}_{t}$, and $z\in
V_{j}^\varepsilon$. Thus, the differential
$d_{x}(\varpi^{n}_{m}\mE')$ of the map $\varpi^{n}_{m}\mE'$ is
injective for any $t\in (0,\min_{j=1,\dots,n}t_{j})$, $M'\in
\mathbb{M}_{t}$ and any $x'\in M'$ such that
$\mE'x'\in\mathtt{Q}_{\varepsilon}$. This proves Lemma \ref{L2},
{\rm ii}.
\end{proof}
\noindent$\bullet \ $ Recall that $r_{AB}$ is defined in (\ref{Eq
def d Hauss}). As a consequence of (\ref{main result on
compacts}), the following holds.
 \begin{corollary}\label{Cor}
The relation
\begin{equation}
\label{main result 1st part} \sup_{M'\in
\mathbb{M}_{t}}\,r_{\mE'(M')\,\mE(M)}\underset{t\to
0}{\longrightarrow}0
\end{equation}
is valid.
 \end{corollary}
\begin{proof} Fix an arbitrary $\varepsilon>0$. In view of (\ref{Eq Close
bound images}), we have $\sup_{M'\in
\mathbb{M}_{t}}d_{H}(\mE'(\Gamma),\mE(\Gamma))=O(t)$. So, for all
sufficiently small $t$ and any $M'\in \mathbb{M}_{t}$, each point
$\xi\in\mE(\Gamma)$ is contained in
$\varepsilon/3$-neighborhood in $\mathbb{C}^{n}$ of some point
$\xi'\in\mE'(\Gamma)$. The set
 $$
K=\left\{\, \zeta\in \mE(M) \ | \ {\rm dist}_{\,\mathbb
C^n}\,(\zeta,\mE'(\Gamma))\geqslant \varepsilon/3\,\right\}
 $$
is compact in $\mE(M)$. As was shown before, there exists a
sufficiently small $\varepsilon>0$ such that $K\subset
\mathtt{Q}_{\varepsilon}$. Then (\ref{main result on compacts})
implies that, for all sufficiently small $t$ and for any $M'\in
\mathbb{M}_{t}$, the set $K$ is contained in
$\varepsilon-$neighborhood of $\mE'(M')$ in $\mathbb{C}^n$. Now,
let $\sigma$ be an arbitrary point of $\mE(M)\backslash K$.
By definition of $K$, one has ${\rm dist}_{\,\mathbb
C^n}(\sigma,\mE'(\Gamma))<\varepsilon/3$. So, $\sigma$ is
contained in $\varepsilon/3$-neighborhood in $\mathbb{C}^{n}$ of
some point $\xi'\in\mE'(\Gamma)$, while $\xi$ is contained in
$\varepsilon/3$-neighborhood in $\mathbb{C}^{n}$ of some point
$\xi'\in\mE'(\Gamma)$. Thus, $K$ is contained in
$\varepsilon$-neighborhood of $\mE'(\Gamma)$ in $\mathbb{C}^{n}$.
Therefore, for all sufficiently small $t$ and for any $M'\in
\mathbb{M}_{t}$, the whole surface $\mE(M)$ is contained in
$\varepsilon$-neighborhood of $\mE'(\Gamma)$ in $\mathbb{C}^{n}$.
Since $\varepsilon>0$ is arbitrary, the latter yields (\ref{main
result 1st part}).
\end{proof}


\section{Closeness near the boundary}

It remains to estimate the closeness of the parts of ${\mE}(M)$
and ${\mE}'(M')$, which are contained in a small neighborhood of
${\mE}(\Gamma)$ and ${\mE'}(\Gamma)$. To do
this, a relevant analog of the bijection (\ref{Eq First
bijection}) that relates the points of these parts, is required.
It is provided by the appropriate local coordinates for points of
${\mE}(M)$ and ${\mE}'(M')$ in such neighborhoods.

\subsubsection*{Coordinates}

$\bullet \ $ Recall that $\mE=\{w_1(\cdot),\dots,w_n(\cdot)\}$ and
$\eta_j=w_j|_\Gamma$.

Fix a point $a\in\Gamma$. Since $\mE$ is projective, there exist a
number $j$ and a small enough open disk $U_{a}\subset\mathbb{C}$
with center $z_{a}:=\eta_{j}(a)$ such that
 \begin{enumerate}
\item $\partial_{\gamma}\eta_{j}(a)\ne 0$ holds,

\item $U_{a}\backslash w_{j}(\Gamma)$ consists of two connected
components $U_{a,0}$ and, $U_{a,1}$, each component being
diffeomorphic to an (open) half-disk, whereas $U_{a,0}\cap
w_{j}(M)=\varnothing$ holds,

\item for any $z\in U_{a,1}$, the function $w_{j}-z$ has a simple zero on $M$ and has no more zeroes on $M$.
\end{enumerate}
Introduce the cylinder
 $$
{\Pi}_{a}:=\underset{1}{\mathbb{C}}\times\dots\times
\underset{j}{U_{a}}\times\dots\times\underset{n}{\mathbb{C}}.
 $$
In view of the properties above, the projection $\pi_{j}$ is a
bijection  from ${\mE}(M)\cap{\Pi}_{a}$ onto $U_{a}\backslash U_{a,0}$. Due
to this, we can regard $z:=\pi_{j}p$ as a {\it local coordinate}
of the point $p\in {\mE}(M)\cap{\Pi}_{a}$.
\smallskip

\noindent$\bullet$\,\,\,Since $\partial_{\gamma}\eta_{j}(a)\ne 0$,
we can choose the neighborhood $\Gamma_{a}$ of $a$ in $\Gamma$
sufficiently small to obey
\begin{equation}
\label{curvestraight1}
\Re\,\frac{\partial_{\gamma}\eta_{j}(l)}{\partial_{\gamma}\eta_{j}(a)}\in
[c_{0},1/c_{0}], \qquad l\in\Gamma_{a}
\end{equation}
with $c_{0}>0$. Without loss of generality, we assume that
$U_{a}\cap w_{j}(\Gamma)=w_{j}(\Gamma_{a})$. Denote
 $$
{\zeta}(z):=\frac{z-z_{a}}{\partial_{\gamma}\eta_{j}(a)}\,,\qquad
\quad {\zeta}_{1}(z):=\Re{\zeta}(z),\quad
{\zeta}_{2}(z):=\Im{\zeta}(z),\quad z\in U_{a}
 $$
and
 $$
\psi={\zeta}\circ\eta_{j}, \quad \psi_{1}:=\Re\psi, \quad
\psi_{2}:=\Im\psi\qquad {\rm on}\,\,\Gamma_a.
 $$
Due to (\ref{curvestraight1}), the map $\psi_{1}: \
\Gamma_{a}\mapsto \mathbb{R}$ is an injection and
${\partial\psi_{1}/\partial l}\in [c_{0},1/c_{0}]$ holds for $l\in
\Gamma_{a}$.

Introduce the new local coordinates
\begin{equation*}
\omega(z):=(s(z),r(z)); \quad
s(z):=[\psi_{1}^{-1}\circ\,{\zeta}_{1}](z), \quad
r(z):={\zeta}_{2}(z)-[\psi_{2}\circ\psi_{1}^{-1}\circ{\zeta}_{1}](z),\qquad
z\in U_{a},
\end{equation*}
then the Jacobian $\mathscr J$ of the passage $(\Re
z,\Im z)\mapsto (s,r)$ obeys
 $$
|\mathscr J(z)|\in [c_{1},1/c_{1}], \qquad z\in U_{a}
 $$
with $c_{1}>0$. Thus, the map $\Omega:={\omega}\circ{\pi}_{j}$ is
a bijection from ${\mE}(M)\cap{\Pi}_{a}$ onto some neighborhood of
$\Gamma_{a}\times\{0\}$ in the strip
$\Gamma_{a}\times[0,+\infty)$. Note that $\Omega({\mE}(l))=(l,0)$
for any $l\in\Gamma_{a}$, so that, in a sense, the passage to
$(s,r)$ rectifies a portion of the near-boundary part of
$\mE(\Gamma_a)$.

\subsubsection*{Primed coordinates}

\noindent$\bullet$\,\,\, Recall that
$\mE'=\{w'_1(\cdot),\dots,w'_n(\cdot)\}$ and
$\eta'_j=w'_j|_\Gamma=\beta'_\Gamma\eta_j$.

Now, let us introduce analogous coordinates on ${\mE}'(M')$ near
$\mE'(\Gamma)$. Let $V_{a}$ be a closed disk {in $\mathbb{C}$}
centered at $z_{a}$ provided {$V_a\subset U_a$}. From
(\ref{curvestraight1}) and estimate (\ref{boundary curves estimate
gen}) it follows that
\begin{equation}
\label{partcurve} \eta'_{j}(\Gamma\backslash\Gamma_{a})\cap
V_{a}=\varnothing,\qquad
|\partial_{\gamma}\eta'_{j}(a)|>|\partial_{\gamma}\eta_{j}(a)|/2>0,
\end{equation}
and
\begin{equation}
\label{curvestraight2}
\Re\,\frac{\partial_{\gamma}\eta'_{j}(l)}{\partial_{\gamma}\eta'_{j}(a)}\in
[c_{0}/2,2/c_{0}], \qquad l\in\Gamma_{a}
\end{equation}
is valid for sufficiently small $t$. By (\ref{boundary
curves estimate gen}), the point $z'_{a}:=\eta'_{j}(a)$ is
close to $z_a=\eta_j(a)$. For $z'\in\mathbb{C}$, denote
\begin{equation*}
{\zeta}'(z'):=\frac{z'-z'_{a}}{\partial_{\gamma}\eta'_{j}(a)},
\qquad  {\zeta}'_{1}(z'):=\Re{\zeta}'(z'), \quad
{\zeta}'_{2}(z'):=\Im{\zeta}'(z'),
\end{equation*}
and
 $$
\psi'={\zeta}'\circ\eta'_{j}, \quad \psi'_{1}:=\Re\psi', \quad \psi'_{2}:=\Im\psi'.
 $$
Due to (\ref{curvestraight2}), the map $\psi'_{1}: \
\Gamma_{a}\mapsto \mathbb{R}$ is an injection and $\partial\psi'_{1}/\partial l\in [c_{0}/2,2/c_{0}]$ for $l\in
\Gamma_{a}$. This and (\ref{partcurve}) imply that the curve
$\eta'_{j}(\Gamma)$ divides $V_{a}$ into two components $V'_{a,0}$
and $V'_{a,1}$, whose closures are diffeomorphic to a half-disk.
\smallskip

\noindent$\bullet$\,\,\, Choose ${z}_{0}\in U_{a,0}\cap V_{a}$ and
${z}_{1}\in U_{a,1}\cap V_{a}$. Due to (\ref{boundary curves
estimate gen}), any fixed neighborhood of the set
${\eta_{j}}(\Gamma_{a})\cap V_{a}$ in $\mathbb{C}$ contains
{$\eta_{j}'(\Gamma_{a})\cap V_{a}$} for sufficiently small $t$.
Hence,
\begin{equation}
\label{dnomcur} {\rm dist}_{\,\mathbb C}
\left({z}_{k},\eta'_{j}(\Gamma_{a})\right)\geqslant c_{2}, \qquad
k=0,1
\end{equation}
holds with $c_{2}>0$. Here and in the subsequent, all the
constants are independent of $M'\in\mathbb{M}_{t}$ and
(sufficiently small) $t$. Thus, we can assume that ${z}_{0}\in
V'_{a,0}$ and ${z}_{1}\in V'_{a,1}$ is fulfilled. In view of the
argument principle and formulas (\ref{dnomcur}),
(\ref{partcurve}), and (\ref{boundary curves estimate gen}), we
have
 $$
m(w'_{j},{z}_{k})=\frac{1}{2\pi
i}\int_{\Gamma}\frac{\partial_{\gamma}\eta'_{j}}{\eta'_{j}-{z}_{k}}\,dl\stackrel{t\to
0}{\longrightarrow}\frac{1}{2\pi
i}\int_{\Gamma}\frac{\partial_{\gamma}\eta_{j}}{\eta_{j}-{z}_{k}}\,dl=m(w_{j},{z}_{k}),
\qquad k=0,1,
 $$
where $m(w,z)$ is the total multiplicity of zeroes of a function
$w-z$ holomorphic on a Riemann surface. Since $m(w,z)$ is integer,
$m(w'_{j},{z}_{k})=m(w_{j},{z}_{k})$ holds for sufficiently small
$t$. Since ${z}_{k}\in U_{a,k}$, we have $m(w_{j},{z}_{k})=k$,
whence $m(w'_{j},{z}_{k})=k$. Note that the integral
 $$
m(w'_{j},z')=\frac{1}{2\pi
i}\int_{\Gamma}\frac{\partial_{\gamma}\eta'_{j}}{\eta'_{j}-z'}\,dl
 $$
is the winding number of the curve $l\mapsto \eta'_{j}(l)$
($l\in\Gamma$) around the point $z'$ and, hence, $m(w'_{j},\cdot)$
is a constant on each connected component of $\mathbb{C}\backslash
\eta'_{j}(\Gamma)$. Thus, $m(w'_{j},z')=k$ for any $z'\in
V'_{a,k}$ i.e. $w'_{j}{(M')}\cap V'_{a,0}=\varnothing$ and the
pre-image ${w_j'}^{\,-1}(z')$ of any $z'\in V'_{a,1}$ consists of
a single point on $\mE'(M)$.
\smallskip

\noindent$\bullet$\,\,\,Introduce the cylinder
 $$
{\Pi}'_{a}:=\underset{1}{\mathbb{C}}\times\dots\times
\underset{j}{V_{a}}\times\dots\times\underset{n}{\mathbb{C}}\,.
 $$
Then the projection ${\pi}_{j}$ is a bijection from
${\mE}'(M')\cap{\Pi}'_{a}$ onto $V_{a}\backslash V'_{a,0}$ and we
can take $z':={\pi}_{j}p$ as a local coordinate of the point $p\in
{\mE}'(M')\cap{\Pi}'_{a}$. Define the new local coordinates
 $$
\omega'(z'):=(s'(z'),r'(z')), \qquad
s'(z'):=[\psi^{'-1}_{1}\circ{\zeta}'_{1}](z'), \qquad
r'(z'):={\zeta}'_{2}(z')-[\psi'_{2}\circ\psi^{'-1}_{1}\circ{\zeta}'_{1}](z').
 $$
Note that $r'(z')\geqslant 0$ for any $z'\in V_{a}\backslash
V'_{a,0}$. In view of (\ref{curvestraight2}) and (\ref{boundary
curves estimate gen}), the function
$\theta'(\tau):=[\psi'_{2}\circ\psi^{'-1}_{1}](\tau)$ obeys
\begin{equation}
\label{thetaest} |\partial\theta'/\partial\tau|\leqslant
c_{3}<\infty, \qquad \tau\in \psi'_{1}(\Gamma_{a}).
\end{equation}
Thus, the Jacobian $\mathscr J'$ of the passage $(\Re z',\Im
z')\mapsto (s',r')$ obeys
 $$|
\mathscr J'(z')|\in [c_{1}/2,2/c_{1}], \qquad z'\in V_{a}
 $$
for sufficiently small $t$. Moreover, estimate (\ref{boundary
curves estimate gen}) implies {that}
 $$
\sup_{z'\in V_{a}}|\omega'(z')-\omega(z')|\underset{t\to
0}\longrightarrow 0
 $$
{and}
 \begin{equation}
\label{coordclose} \sup_{z'\in V_{a}}|\,[\omega^{-1}\circ\,
\omega'](z')-z'\,|\underset{t\to 0}\longrightarrow 0
 \end{equation}
hold uniformly with respect to $M'\in \mathbb{M}_{t}$.
Thus, for sufficiently small $t$ we have $\omega'(V_{a}\backslash
V'_{a,0})\subset \omega(U_{a})\cap
\left[\,\Gamma_{a}\times[0,+\infty)\,\right]=\omega(U_{a}\backslash
U_{a,0})$. Therefore, the map $\Omega':=\omega'\circ{\pi}_{j}$ is
an injection from ${\mE}'(M')\cap{\Pi}'_{a}$ into
$\omega(U_{a}\backslash U_{a,0})=\Omega({\mE}(M)\cap{\Pi}_{a})$.
Note that $\Omega'({\mE}'(l))=(l,0)$ for any $l\in\Gamma$ such
that ${\mE}'(l)\in{\Pi}'_{a}$. So, for small $t$, the map
 $$
\Omega^{-1}\circ \,\Omega': \
{\mE}'(M')\cap{\Pi}'_{a}\mapsto{\mE}(M)\cap{\Pi}_{a}
 $$
is well-defined and is an injection.

With each $p'\in{\mE}'(M')\cap{\Pi}'_{a}$ we associate the point
$p:=[\Omega^{-1}\circ \Omega']\,(p')$. It is the required local
{bijection} between near-boundary points of $\mE(M)$ and
$\mE'(M')$. Note that, if $p'={\mE}'(l)$ with $l\in\Gamma$, then
$p={\mE}(l)$.

\subsubsection*{Estimates via argument principle}
The next step is to establish the closeness of points $p'$ and
$p$. From a technical point of view, the difficulty consists in
estimating the values of Cauchy-type integrals {near the
integration contour.}
 \begin{lemma}\label{L last}
\

{\rm i.} The relation
 \begin{equation}\label{proximitineartboundary}
\sup_{M'\in\mathbb{M}_{t}}\Big(\sup_{p'\in
{\mE}'(M')\cap{\Pi'}_{a}}\|\,[\Omega^{-1}\circ
\Omega'](p')-p'\,\|_{\mathbb C^n}\Big)\underset{t\to
0}\longrightarrow 0
 \end{equation}
is valid.

{\rm ii}. Suppose that the map $\varpi^{n}_{m}\mE$, $m<n$ is an immersion {\rm (}the projection $\varpi^{n}_{m}$ is given by {\rm (\ref{projection le s})}{\rm )}. Then, for sufficiently small $t>0$ and any $M'\in \mathbb{M}_{t}$, the map $\varpi^{n}_{m}\mE': \ \mE^{'-1}(\Pi_{a})\mapsto \mathbb{C}^{m}$ is an immersion.
 \end{lemma}
\begin{proof}\,\,\,
It suffices to prove (\ref{proximitineartboundary}) for
${\mE}'(M')$ replaced by ${\mE}'(M'\backslash\Gamma)$. This fact
follows from (\ref{Eq Close bound images}) since
$|\,[\Omega^{-1}\circ \Omega'](p')-p'\,|=|{\mE}(l)-{\mE}'(l)|$ for
any $p'={\mE}'(l)$.
\smallskip

\noindent$\bf 1.$\,\,\,Let $p'$ be an arbitrary point of
${\mE}'(M'\backslash\Gamma)\cap{\Pi}'_{a}$. Denote
$p:=[\Omega^{-1}\circ \,\Omega']\,(p')$, $z':={\pi}_{j}p'$,
$z:={\pi}_{j}p$, and $(s,r):=\Omega'(p')=\Omega(p)$. Note that
$z'\in V'_{a,1}$, $z\in U_{a,1}$, and $r>0$. Hence, the pre-image
{${w'_j}^{\,-1}(\{z'\})$} contains a single point $x'\in M'$,
whereas {${w_j}^{-1}(\{z\})$} consists of a single point $x\in M$.
Also, $x'$ is a simple zero of $w'_{j}-z'$ and $x$ is a simple
zero of $w_{j}-z$. By the generalized argument principle, we have
\begin{equation}
\label{argprincex}
\begin{split}
{\pi}_{k}p'=w'_{k}(x')=\frac{1}{2\pi
i}\int_{\Gamma}\frac{\eta'_{k}\,\partial_{\gamma}\eta'_{j}}{\eta'_{j}-z'}\,dl=
\frac{1}{2\pi i}\int_{\Gamma}\frac{(\eta'_{k}-\eta'_{k}(s))\,\partial_{\gamma}\,\eta'_{j}}{\eta'_{j}-z'}\,dl+\\
+\eta'_{k}(s)\,\frac{1}{2\pi
i}\,\int_{\Gamma}\frac{\partial_{\gamma}\eta'_{j}}{\eta'_{j}-z'}\,dl=
\frac{1}{2\pi
i}\,\int_{\Gamma}\frac{(\eta'_{k}-\eta'_{k}(s))\,\partial_{\gamma}\eta'_{j}}{\eta'_{j}-z'}\,dl+\eta'_{k}(s).
\end{split}
\end{equation}
For $\delta>0$, denote $\Gamma_{\delta}(s)=\{l\in\Gamma \ | \ {\rm
dist}_{\Gamma}(l,s)\leqslant \delta\}$ (then
$\Gamma_{\delta}(s)\subset\Gamma_{a}$ for sufficiently small $t$
and $\delta$). Put
\begin{equation}
\label{integargparts} J'_{k,\delta}(z'):=\frac{1}{2\pi
i}\int_{\Gamma_{\delta}}\frac{(\eta'_{k}-\eta'_{k}(s))\,\partial_{\gamma}\eta'_{j}}{\eta'_{j}-z'}\,dl,
\qquad \tilde{J}'_{k,\delta}(z'):=\frac{1}{2\pi
i}\int_{\Gamma\backslash\Gamma_{\delta}}\frac{(\eta'_{k}-\eta'_{k}(s))\,\partial_{\gamma}\eta'_{j}}{\eta'_{j}-z'}\,dl.
\end{equation}
\smallskip

\noindent$\bf 2.$\,\,\,Let us estimate $|J'_{k,\delta}(z')|$. We
have
 $$
J'_{k,\delta}(z')=\frac{1}{2\pi
i}\int_{\Gamma_{\delta}}\frac{\eta'_{k}-\eta'_{k}(s)}{\eta'_{j}-\eta'_{j}(s)}\cdot\frac{\eta'_{j}-\eta'_{j}(s)}{\eta'_{j}-z'}\,\partial_{\gamma}\eta'_{j}\,dl.
 $$
Recall that
$\psi'={\zeta}^{'}\circ\eta'_{j}=(\eta'_{j}-z_{0}')/\partial_{\gamma}\eta'_{j}(a)$
and $s=[\psi^{'-1}_{1}\circ\,{\zeta}'_{1}](z')$. Hence,
\begin{align*}
\Big|\frac{\eta'_{j}(l)-\eta'_{j}(s)}{\eta'_{j}(l)-z'}\Big|=
\Big|\frac{[{\zeta}^{'-1}\circ\psi'](l)-[{\zeta}^{'-1}\circ\psi'](s)}{[{\zeta}^{'-1}\circ\psi'](l)-
z'}\Big|=\Big|\frac{\psi'(l)-\psi'(s)}{\psi'(l)-{\zeta}^{'}(z')}\Big|\leqslant \\
\leqslant
\Big|\frac{\psi'(l)-\psi'(s)}{\psi'_{1}(l)-\psi'_{1}(s)}\Big|\leqslant
1+
\Big|\frac{\psi'_{2}(l)-\psi'_{2}(s)}{\psi'_{1}(l)-\psi'_{1}(s)}\Big|=
1+\Big|\frac{\theta'(\tau_{l})-\theta'(\tau_{s})}{\tau_{l}-\tau_{s}}\Big|
\end{align*}
holds, where $\tau_{l}:=\psi'_{1}(l)$, $\tau_{s}:=\psi'_{1}(s)$,
and $\theta'=\psi'_{2}\circ{\psi'_1}^{\,-1}$. Hence, in view of
(\ref{thetaest}), we get
\begin{equation}
\label{trianglelemma}
\Big|\frac{\eta'_{j}(l)-\eta'_{j}(s)}{\eta'_{j}(l)-z'}\Big|\leqslant
1+c_{3}.
\end{equation}
Relations (\ref{partcurve}) and (\ref{curvestraight2}) imply
\begin{align*}
|\eta'_{j}(l)-\eta'_{j}(s)|\geqslant
|\partial_{\gamma}\eta'_{j}(a)|\Big|\Re\frac{\eta'_{j}(l)-
\eta'_{j}(s)}{\partial_{\gamma}\eta'_{j}(a)}\Big|=|\partial_{\gamma}\eta'_{j}(a)
|\Big|\Re\frac{\eta'_{j}(l)}{\partial_{\gamma}\eta'_{j}(a)}-
\Re\frac{\eta'_{j}(s)}{\partial_{\gamma}\eta'_{j}(a)}\Big|=\\
=|\partial_{\gamma}\eta'_{j}(a)|\,\Big|\int\limits_{s}^{l}
\Re\frac{\partial_{\tau}\eta'_{j}(\tau)}{\partial_{\gamma}\eta'_{j}(a)}\,d\tau\Big|\geqslant
\frac{|\partial_{\gamma}\eta_{j}(a)|\,c_{0}}{4}\,{\rm
dist}_{\Gamma}(s,l).
\end{align*}
Therefore, in view of (\ref{boundary curves estimate gen}),
\begin{equation}
\label{kviaj}
\Big|\frac{\eta'_{k}-\eta'_{k}(s)}{\eta'_{j}-\eta'_{j}(s)}\Big|\leqslant
\frac{4}{|\partial_{\gamma}\eta_{j}(a)|\,c_{0}}
\parallel\eta'_{k}\parallel_{C^{1}(\Gamma;\,\mathbb C)}\leqslant \frac{8}{|\partial_{\gamma}\eta_{j}(a)|\,c_{0}}\parallel\eta_{k}\parallel_{C^{1}(\Gamma;\,\mathbb C)}
\end{equation}
for sufficiently small $t$. Combining (\ref{trianglelemma}),
(\ref{kviaj}), and (\ref{boundary curves estimate gen}), we obtain
\begin{equation}
\label{badintest} |J'_{k,\delta}(z')|\leqslant c_{4}\,\delta\,.
\end{equation}
\smallskip

\noindent$\bf 3.$\,\,\, Let $J_{k,\delta}(k,z)$ and
$\tilde{J}_{k,\delta}(k,z)$ be defined by formula
(\ref{integargparts}) with omitted primes. The same arguments as
above show that we can omit primes in (\ref{argprincex}),
(\ref{badintest}). In particular,
\begin{equation}
\label{expproject}
{\pi}_{k}p'-{\pi}_{k}p=J'_{k,\delta}(z')-J_{k,\delta}(z)+\tilde{J}'_{k,\delta}(z')-\tilde{J}_{k,\delta}(z)+\eta'_{k}(s)-\eta_{k}(s).
\end{equation}
Fix an arbitrary $\varepsilon>0$ and take
$\delta=\delta(\varepsilon)\in(0,\varepsilon/6c_{4})$, then
(\ref{badintest}) yields
$|J'_{k,\delta}(z')|+|J_{k,\delta}(z)|\leqslant\varepsilon/3 $.
Now, let us estimate
$|\tilde{J}'_{k,\delta}(z')-\tilde{J}_{k,\delta}(z)|$. We have
\begin{equation}
\label{dfghjhj}
|\tilde{J}'_{k,\delta}(z')-\tilde{J}_{k,\delta}(z)|=\Big|\int_{\Gamma\backslash\Gamma_{\delta}}
\frac{(\eta'_{k}-\eta'_{k}(s))\,(\eta_{j}-z)\,\partial_{\gamma}\eta'_{j}-
(\eta_{k}-\eta_{k}(s))\,(\eta'_{j}-z')\,\partial_{\gamma}\eta_{j}}{2\pi(\eta'_{j}-z')(\eta_{j}-z)}\,\,dl\,\Big|.
\end{equation}
Recall that ${\partial\psi'_{1}/\partial l}\in [c_{0}/2,2/c_{0}]$
for $l\in \Gamma_{a}$ and
$|\partial_{\gamma}\eta'_{j}(a)|\geqslant
|\partial_{\gamma}\eta_{j}(a)|{/2}$. Hence,
\begin{align*}
|\eta'_{j}(l)-z'|=|{\zeta}^{'-1}\circ\psi'(l)-z'|=|\partial_{\gamma}\eta'_{j}(a)|\,|\psi'(l)-
{\zeta}'(z')|\geqslant |\partial_{\gamma}\eta'_{j}(a)||\psi'_{1}(l)-{\zeta}'_{1}(z')|=\\
=|\partial_{\gamma}\eta'_{j}(a)||\psi'_{1}(l)-\psi'_{1}(s)|\geqslant
|\partial_{\gamma}
\eta'_{j}(a)|\,\Big|\int_{s}^{l}\partial_{\tau}\psi'_{1}(\tau)\,d\tau\Big|\geqslant
\frac{c_{0}}{4}\,|\partial_{\gamma}\eta_{j}(a)|\,\delta.
\end{align*}
Similarly, we obtain $|{\eta_{j}(l)-z}|\geqslant
c_{0}\,|\partial_{\gamma}\eta_{j}(a)|\,\delta/4$. Thus, the
denominator in (\ref{dfghjhj}) does not become small. Also, the
$|z-z'|=|\,[{\omega\circ {\omega'}^{\,-1}}]\,(z')-z'|$ obeys
(\ref{coordclose}). So, the same arguments as used in the proof of
(\ref{boundary integral estimates}) show that
\begin{equation}
\label{secintest} \sup_{M'\in\mathbb{M}_{t}}\sup_{p'\in {\mE}'
(M'\backslash\Gamma)\cap
{\Pi'}_{a}}|\tilde{J}'_{k,\delta}(z')-\tilde{J}_{k,\delta}(z)|\underset{t\to
0}\longrightarrow 0.
\end{equation}
Choose $t(\varepsilon)>0$ such that the left-hand sides of
(\ref{secintest}) and (\ref{boundary curves estimate gen}) are
less than $\varepsilon/3$ for any $t\in(0,t(\varepsilon))$. Now,
(\ref{expproject}) implies that
$|{\pi}_{k}p'-{\pi}_{k}p|<\varepsilon$ for any $t\in(0,t_{0})$,
$M'\in\mathbb{M}_{t}$, and $p'\in
{\mE}'(M'\backslash\Gamma)\cap{\Pi'}_{a}$. Since $\varepsilon>0$
and $k=1,\dots,n$ are arbitrary, formula
(\ref{proximitineartboundary}) is proved.
\smallskip

\noindent$\bf 4.$\,\,\, Differentiating (\ref{argprincex}), we obtain
\begin{equation*}
\partial_{z'}w_{k}'(x')=\frac{1}{2\pi
i}\int\limits_{\Gamma}\frac{\eta'_{k}\,\partial_{\gamma}\eta'_{j}}{(\eta'_{j}-z')^{2}}\,dl=\mathfrak{J}'_{k,\delta}(z')+\tilde{\mathfrak{J}}'_{k,\delta}(z')+\mathfrak{K}'_{k}(z')+\tilde{\mathfrak{K}}'_{k}(z'),
\end{equation*}
where
\begin{equation}
\label{argprincex diff}
\begin{split}
\mathfrak{J}'_{k,\delta}(z'):&=\frac{1}{2\pi i}\int\limits_{\Gamma_{\delta}}\frac{\big(\eta'_{k}(l)-\eta'_{k}(s)-\partial_{\gamma}\eta'_{k}(s)\partial_{\gamma}\eta'_{j}(s)^{-1}(\eta'_{j}(l)-\eta'_{j}(s)\big)\partial_{\gamma}\eta'_{j}}{(\eta'_{j}-z')^{2}}dl,\\
\tilde{\mathfrak{J}}'_{k,\delta}(z'):&=\frac{1}{2\pi i}\int\limits_{\Gamma\backslash\Gamma_{\delta}}\frac{\big(\eta'_{k}(l)-\eta'_{k}(s)-\partial_{\gamma}\eta'_{k}(s)\partial_{\gamma}\eta'_{j}(s)^{-1}(\eta'_{j}(l)-\eta'_{j}(s)\big)\partial_{\gamma}\eta'_{j}}{(\eta'_{j}-z')^{2}}dl,\\
\mathfrak{K}'_{k}(z'):&=\frac{\eta'_{k}(s)}{2\pi i}\int\limits_{\Gamma}\frac{\partial_{\gamma}\eta'_{j}}{(\eta'_{j}-z')^{2}}dl=\frac{\eta'_{k}(s)}{2\pi i}\int\limits_{\eta_{j}(\Gamma)}\frac{dz''}{(z''-z')^{2}}=0\\
\tilde{\mathfrak{K}}'_{k}(z'):&=\frac{\partial_{\gamma}\eta'_{k}(s)}{\partial_{\gamma}\eta'_{j}(s)}\frac{1}{2\pi i}\int\limits_{\Gamma}\frac{(\eta'_{j}(l)-\eta'_{j}(s))\partial_{\gamma}\eta'_{j}}{(\eta'_{j}-z')^{2}}dl=\frac{\partial_{\gamma}\eta'_{k}(s)}{\partial_{\gamma}\eta'_{j}(s)}\Big(\frac{1}{2\pi i}\int\limits_{\Gamma}\frac{(\eta'_{j}(l)-z')\partial_{\gamma}\eta'_{j}}{(\eta'_{j}-z')^{2}}dl+\\
+&\frac{z'-\eta'_{j}(s)}{2\pi i}\int\limits_{\Gamma}\frac{\partial_{\gamma}\eta'_{j}}{(\eta'_{j}-z')^{2}}dl\Big)=\frac{\partial_{\gamma}\eta'_{k}(s)}{\partial_{\gamma}\eta'_{j}(s)}
\end{split}
\end{equation}
Similarly,
\begin{equation*}
\partial_{z}w_{k}(x)=\frac{1}{2\pi
i}\int\limits_{\Gamma}\frac{\eta_{k}\,\partial_{\gamma}\eta_{j}}{(\eta_{j}-z)^{2}}\,dl=\mathfrak{J}_{k,\delta}(z)+\tilde{\mathfrak{J}}_{k,\delta}(z)+\mathfrak{K}_{k}(z)+\tilde{\mathfrak{K}}_{k}(z),
\end{equation*}
where the summands on the right are defined by (\ref{argprincex diff}) with omitted primes. In view of (\ref{boundary curves estimate gen}) and (\ref{partcurve}),
$$\mathfrak{F}_{k}(s,l):=\big|\eta'_{k}(l)-\eta'_{k}(s)-\frac{\partial_{\gamma}\eta'_{k}(s)}{\partial_{\gamma}\eta'_{j}(s)}\big(\eta'_{j}(l)-\eta'_{j}(s)\big)\big|\le c\sum_{p=1}^{n}\parallel\eta'_{p}\parallel_{C^{2}(\Gamma;\mathbb{C})}({\rm dist}_{\Gamma}\{s,l\})^{2}\le ct({\rm dist}_{\Gamma}\{s,l\})^{2}.$$
Hence, formulas (\ref{trianglelemma}) and (\ref{boundary curves estimate gen}) imply
\begin{align}
\label{ergeger}
\mathfrak{J}'_{k,\delta}(z'):=\frac{1}{2\pi}\int\limits_{\Gamma_{\delta}}\Big|\frac{\mathfrak{F}_{k}(s,l)}{(\eta'_{j}(l)-\eta'_{j}(s))^{2}}\Big|\Big|\frac{(\eta'_{j}(l)-\eta'_{j}(s))}{(\eta'_{j}-z')}\Big|^{2} |\partial_{\gamma}\eta'_{j}(l)|d\gamma\underset{\delta\to
0}\longrightarrow 0.
\end{align}
uniformly with respect to $M'\in\mathbb{M}_{t}$ and $x'\in M'$ such that $\mE'(x')\in {\mE}'(M'\backslash\Gamma)\cap{\Pi'}_{a}$. Note that (\ref{ergeger}) remains valid with omitted primes. For any fixed $\delta>0$, the formula
\begin{equation*}
\sup_{M'\in\mathbb{M}_{t}}\sup_{p'\in {\mE}'
(M'\backslash\Gamma)\cap
{\Pi'}_{a}}|\tilde{\mathfrak{J}}'_{k,\delta}(z')-\tilde{\mathfrak{J}}_{k,\delta}(z)|\underset{t\to
0}\longrightarrow 0
\end{equation*}
is valid and it is obtained by repeating the arguments leading to (\ref{secintest}). Also, for any fixed $\delta>0$, the formula
\begin{equation*}
\sup_{M'\in\mathbb{M}_{t}}\sup_{p'\in {\mE}'
(M'\backslash\Gamma)\cap
{\Pi'}_{a}}|\tilde{\mathfrak{K}}'_{k}(z')-\tilde{\mathfrak{K}}_{k}(z)|\underset{t\to
0}\longrightarrow 0
\end{equation*}
follows from (\ref{boundary curves estimate gen}) and (\ref{curvestraight1}). Thus,
$$|\partial_{z'}w_{k}'(x')-\partial_{z}w_{k}(x)|\underset{t\to 0}\longrightarrow 0,$$
uniformly with respect to $M'\in\mathbb{M}_{t}$, and any $x'\in M'$ such that $\mE'(x')\in {\mE}'(M'\backslash\Gamma)\cap{\Pi'}_{a}$. Now, arguing as in the proof of Lemma \ref{L2}, ii, we prove the statement ii, Lemma \ref{L last}.
\end{proof}

\subsubsection*{Completing the proof of Theorem \ref{MT}}

\noindent$\bullet$\,\,\,As a corollary of (\ref{proximitineartboundary}), we
obtain
 \begin{equation}\label{proximitineartboundary1}
\sup_{M'\in\mathbb{M}_{t}}\Big(\sup_{p'\in
{\mE}'(M')\cap{\Pi'}_{a}} {\rm
dist}_{\mathbb{C}^{n}}(p',{\mE'}(M)\cap{{\Pi}_{a}})\Big)\underset{t\to
0}\longrightarrow 0.
 \end{equation}
for any $a\in\Gamma$. The cylinders
${\Pi'}_{a}\backslash\partial{\Pi'}_{a}$, $a\in\Gamma$ constitute
an open cover of the compact set ${\mE}(\Gamma)$ in
$\mathbb{C}^{n}$. Choose a finite subcover
${\Pi'}_{a_{i}}\backslash\partial{\Pi'}_{a_{i}}$, $i=1,\dots,N$
and denote
{$\tilde{\mathtt{Q}'}:=\bigcup_{i=1}^{N}{\Pi'}_{a_{i}}\backslash\partial{\Pi'}_{a_{i}}$
and $\tilde{\mathtt{Q}}:=\bigcup_{i=1}^{N}{\Pi}_{a_{i}}$}. Then
(\ref{proximitineartboundary1}) implies {
\begin{equation}\label{nearboundary convergence}
\sup_{M'\in\mathbb{M}_{t}}r_{[{\mE}(M)\cap\overline{\tilde{\mathtt{Q}}}\,]\,\,[{\mE}'(M')\cap\overline{\tilde{\mathtt{Q}'}}\,]}=
\sup_{M'\in\mathbb{M}_{t}}\Big(\sup_{p'\in
{\mE}'(M')\cap\tilde{\mathtt{Q}'}}{\rm
dist}_{\mathbb{C}^{n}}(p',{\mE}(M)\cap\tilde{\mathtt{Q}})\Big)\underset{t\to
0}\longrightarrow 0.
 \end{equation}}
Note that {${\mE}(M)\backslash\tilde{\mathtt{Q}}'$} is a compact
subset of ${\mE}(M\backslash\Gamma)$; as shown above, it is
contained in $\mathtt{Q}_{\varepsilon}$ for sufficiently small
$\varepsilon>0$. Denote
$\mathtt{Q}=\mathtt{Q}_{\varepsilon}\cup\tilde{\mathtt{Q}}$, then
${\mE}(M)\subset \mathtt{Q}$. Formulas (\ref{main result on
compacts}) and (\ref{nearboundary convergence}) yield
\begin{equation}
\label{convergence 1}
\sup_{M'\in\mathbb{M}_{t}}r_{{\mE}(M)\,\,[{\mE}'(M')\cap\overline{\mathtt{Q}}\,]}\underset{t\to
0}\longrightarrow 0.
\end{equation}
Denote the $\varepsilon-$neighborhood of ${\mE}(M)$ by
$\mathcal{Q}_{\varepsilon}$. Let $\varepsilon>0$ be sufficiently
small for $\overline{\mathcal{Q}_{\varepsilon}}$ to be contained
in $\mathtt{Q}$. Due to (\ref{convergence 1}), there is
$t(\varepsilon)>0$ such that
${\mE}'(M')\cap\overline{\mathtt{Q}}\subset\mathcal{Q}_{\varepsilon}$
for all $t\in (0,t(\varepsilon))$ and $M'\in\mathbb{M}_{t}$. Let
us show that ${\mE}'(M')\subset\mathtt{Q}$ holds for all
$t\in (0,t(\varepsilon)]$ and $M'\in\mathbb{M}_{t}$. Assume
the opposite, then there exists $M'\in\mathbb{M}_{t(\varepsilon)}$
and $p'\in {\mE}'(M')\backslash\mathtt{Q}$. Since ${\mE}'(M')$ is
connected,  there exists a path $\mathcal{L}$ in ${\mE}'(M')$ with
the beginning at $p'$ and the end at some point $p''\in
{\mE}'(\Gamma)\subset{\mE}'(M')\cap\mathtt{Q}\subset\mathcal{Q}_{\varepsilon}$.
Hence, the set $\mathcal{L}\cap
(\mathtt{Q}\backslash\mathcal{Q}_{\varepsilon})$ is nonempty.
However,
$\mathcal{L}\cap(\mathtt{Q}\backslash\mathcal{Q}_{\varepsilon})
\subset{\mE}'(M')\cap(\mathtt{Q}\backslash\mathcal{Q}_{\varepsilon})=\varnothing$
is valid. Thus, our assumption has led to a contradiction.
Therefore ${\mE}'(M')\subset\mathtt{Q}$ holds for
sufficiently small $t$ and formula (\ref{convergence 1}) can be
rewritten as
\begin{equation*}
\sup_{M'\in\mathbb{M}_{t}}r_{{\mE}(M)\,{\mE}'(M')}\underset{t\to
0}\longrightarrow 0.
\end{equation*}
The latter relation along with (\ref{main result 1st part}) and
definition (\ref{Eq def d Hauss}) imply (\ref{main result}) for
any projective immersion $\mE$.  Thereby, as shown in the
beginning of Section \ref{Sec2}, formula (\ref{main result}) is
proved for any immersion $\mE$.

\noindent$\bullet$\,\,\, To complete the proof of (\ref{supinf
main}), it remains to show that the extension
$\mE':M'\mapsto\mathbb\mathbb{C}^{n}$ of the immersion ${\mE}:\
M\mapsto \mathbb{C}^{n}$ is an immersion for sufficiently small
$t>0$ and any $M'\in\mathbb{M}_{t}$. If ${\mE}$ is projective,
then the statement is obvious and $\mE'$ is also projective for
small $t$. If ${\mE}:\ M\mapsto \mathbb{C}^{n}$ is not projective,
then it can be completed to a projective immersion. In this case,
the statement follows from from Lemma \ref{L2}, ii, and Lemma
\ref{L last}, ii.

{\it

Theorem \ref{MT} is proved.}

\smallskip

$\bullet \ $ Perhaps, relation (\ref{main result}) may be
improved to an estimate $d_{H}(\mE'(M'),\mE(M))\underset{t\to 0}=
O(t^\alpha)$ with a positive $\alpha$. However, it requires more
subtle considerations `near boundary'.

\subsection*{Statements and Declarations}
\paragraph{Competing Interests.} On behalf of all authors, the corresponding author states that there is no conflict of interest.

\paragraph{Data Availibility Statement.} Data sharing not applicable to this article as no datasets were generated oranalysed during the current study.

\paragraph{Funding.} M. I. Belishev and D. V. Korikov were supported by the RFBR grant
20-01-00627-a.

\end{document}